%% file: arxiv.tex
\documentclass[superscriptaddress,aps,prl,nofootinbib,reprint]{revtex4-2}
\usepackage[T1]{fontenc}
\usepackage[utf8]{inputenc}
\usepackage{color}
\usepackage{babel}
\usepackage{amsmath}
\usepackage{amsthm}
\usepackage{amssymb}
\usepackage{dsfont}
\usepackage{physics}
\usepackage{graphicx}
\usepackage{import}
\usepackage{enumitem}
\usepackage[normalem]{ulem}
\usepackage[unicode=true,pdfusetitle,
bookmarks=true,bookmarksnumbered=false,bookmarksopen=false,
breaklinks=false,pdfborder={0 0 0},pdfborderstyle={},backref=false,colorlinks=true]
{hyperref}
\hypersetup{
allcolors=magenta}

\theoremstyle{plain}

\newtheorem{theorem}{Theorem}

\newtheorem{lemma}{Lemma}

\theoremstyle{plain}

\theoremstyle{plain}

\ifx\proof\undefined
\newenvironment{proof}[1][\protect\proofname]{\par
\normalfont\topsep6\p@\@plus6\p@\relax
\trivlist
\itemindent\parindent
\item[\hskip\labelsep\scshape #1]\ignorespaces
}{%
\endtrivlist\@endpefalse
}
\providecommand{\proofname}{Proof}
\fi

\usepackage{times}
\usepackage{txfonts}
\usepackage{braket}
\usepackage{colortbl}

\makeatother

\providecommand{\propositionname}{Proposition}
\providecommand{\theoremname}{Theorem}
\providecommand{\lemmaname}{Lemma}

\begin{document}
\title{Fundamental limitations on entanglement extraction from purity}

\author{Tulja Varun Kondra$^\ddagger$}
\email{tuljavarun@gmail.com}
\affiliation{Heinrich Heine University D\"{u}sseldorf, Faculty of Mathematics and Natural Sciences, Institute for Theoretical Physics III}

\author{Pedro Barrios Hita$^\ddagger$}
\email{pbhita99@gmail.com}
\affiliation{Institute of Software Technology, German Aerospace Center (DLR), Sankt Augustin, Germany}
\affiliation{Heinrich Heine University D\"{u}sseldorf, Faculty of Mathematics and Natural Sciences, Institute for Theoretical Physics III}
\author{Justus Neumann}
\affiliation{Heinrich Heine University D\"{u}sseldorf, Faculty of Mathematics and Natural Sciences, Institute for Theoretical Physics III}
\author{Hermann Kampermann}
\affiliation{Heinrich Heine University D\"{u}sseldorf, Faculty of Mathematics and Natural Sciences, Institute for Theoretical Physics III}
\author{Dagmar Bruß}
\affiliation{Heinrich Heine University D\"{u}sseldorf, Faculty of Mathematics and Natural Sciences, Institute for Theoretical Physics III}

\begin{abstract}
States of sufficiently low purity are separable and cannot be entangled by unital (purity-non-generating) operations. Since high-purity states are experimentally demanding, it is natural to ask how much purity a state must possess to enable entanglement generation. Absolutely separable states remain separable under all deterministic unital channels, and so cannot deterministically generate entanglement in this setting. We show, however, that some absolutely separable states can generate entanglement via probabilistic protocols that do not produce purity. This motivates the study of states that fail to generate entanglement with \emph{any} non-zero probability, which we call \emph{completely absolutely separable}; we give a full characterization of this class. Along the way, we derive a novel sufficient condition for separability that depends only on the largest and smallest  eigenvalues, along with the smallest local dimension and is independent of all previously known spectral separability criteria. 
\end{abstract}

\maketitle

\def\thefootnote{$\ddagger$}\footnotetext{These authors contributed equally to this work.}

Quantum entanglement is a fundamental resource in quantum information science, playing a key role in tasks such as quantum cryptography \cite{RevModPhys.74.145}, quantum teleportation \cite{bennett1993teleporting} and quantum metrology \cite{PhysRevLett.96.010401}. The investigation of entanglement and its properties has not only deepened our understanding of quantum mechanics, but has also shown how quantum effects can be controlled and exploited for practical applications \cite{RevModPhys.81.865}. In a broader context, quantum resource theories provide a unified mathematical framework for studying a range of other quantum resources, such as coherence \cite{RevModPhys.89.041003}, athermality \cite{Lostaglio_2019}, and magic \cite{Veitch_2014}. Purity is one such resource, quantifying the deviation of a quantum state from the maximally mixed state \cite{Streltsov_2018}. In many experimental platforms, the preparation of high-purity states is itself a demanding task, as unavoidable noise, decoherence, and imperfect control tend to introduce mixing \cite{RevModPhys.75.715,PhysRevLett.93.040501,Krantz_2019}. As a result, the ability to generate and maintain high-purity states is often a significant practical limitation, making it important to understand what can be achieved under restricted purity.

Therefore, a natural question is: \emph{how much purity is required to enable entanglement generation?} To address this question, we consider the most general physical processes that do not generate purity, namely unital channels \cite{Streltsov_2018,Mendl_2009,Gour_2015}. This setting captures the fundamental interplay between purity and entanglement. It is known that states sufficiently close to the maximally mixed state cannot be transformed into entangled states by any unital channel \cite{PhysRevA.66.062311,PhysRevA.75.062330}. This naturally leads to the problem of characterizing those states that are fundamentally incapable of producing entanglement under purity-non-generating dynamics. This set coincides with the class of so-called \emph{absolutely separable} (AS) states \cite{PhysRevA.63.032307,PhysRevA.108.042402}, originally introduced as states that remain separable under all global unitary operations. A complete characterization of such states is known when one of the local dimensions is equal to two \cite{PhysRevA.88.062330}. For higher local dimensions, however, the problem remains open and is known as the longstanding \emph{separability from spectrum} problem \cite{KnillSeparabilitySpectrum}.

Even though a complete characterization is not known, the positive-partial-transpose (PPT) criterion for separability \cite{PhysRevLett.77.1413,Horodecki_1996} can be used to provide necessary conditions in arbitrary local dimensions. These conditions are expressed as linear matrix inequalities and are tight for small local dimensions (in particular, when one of the local dimensions is two). However, for higher local dimensions it remains an open conjecture whether the PPT criterion can detect all absolutely separable states \cite{arunachalam2015absoluteseparabilitydeterminedpartial}. Other variants of necessary conditions can be obtained, for instance, from the reduction criterion \cite{JIVULESCU2015276} and from bounds based on conditional entropy \cite{PhysRevA.96.062102}. On the other hand, various sufficient conditions have also been proposed, which approximate the set of absolutely separable states from below \cite{PhysRevA.66.062311,PhysRevA.75.062330,Abellanet_Vidal_2025,Lewenstein_2016}.

In recent years, there has been growing interest in probabilistic protocols \cite{PhysRevLett.128.110505,Regula2022tightconstraints}. Unlike deterministic transformations, probabilistic protocols allow for a nonzero probability of failure and only require the desired transformation to occur conditionally upon success. This additional flexibility can  enlarge the set of achievable state transformations and, in particular, allows one to overcome limitations that are unavoidable in the deterministic setting \cite{PhysRevA.107.042401,Regula_2024}. These developments indicate that any fundamental constraint on converting purity into entanglement must account for such general probabilistic transformations, rather than only trace-preserving ones. In particular, a purely deterministic notion of absolute separability may fail to capture the strongest possible limitation to entanglement generation from purity.

In this letter, we allow for the most general probabilistic processes that do not generate purity, even when conditioned on any individual outcome; we refer to such processes as \emph{probabilistic unital channels}. We first show that certain absolutely separable states can generate entanglement with nonzero probability under suitable probabilistic unital channels. This demonstrates that absolute separability does not capture the ultimate purity based limitation on entanglement generation. Motivated by this observation, we introduce a stricter notion, which we call \emph{completely absolutely separable} (CAS) states, defined as bipartite states that cannot generate entanglement with any nonzero probability under any probabilistic unital channel. We show that the set of CAS states admits a remarkably simple characterization in terms of necessary and sufficient spectral conditions, valid for arbitrary local dimensions. Furthermore, we establish an interesting connection between this newly introduced notion of CAS and the conventional definition of absolute separability. In particular, we show that CAS states are precisely those states that remain absolutely separable even after attaching (via tensor product) a maximally mixed ancillary system of arbitrary dimension, to one of the subsystems. 

Our analysis yields a new sufficient condition for separability, expressed solely in terms of the largest and smallest eigenvalues of the state and the smallest local dimension. This condition is independent of all previously known sufficient criteria for absolute separability \cite{PhysRevA.66.062311,PhysRevA.75.062330,Abellanet_Vidal_2025,Lewenstein_2016}. To the best of our knowledge, all earlier sufficient conditions depend on the global dimension rather than the smallest local dimension. As a consequence, our criterion certifies absolute separability for states beyond the reach of all previously known sufficient criteria (see Fig.~\ref{fig:fig_1}). Indeed, we explicitly construct (see Sec.~G of the Supplemental Material), for every bipartite system with unequal local dimensions, an absolutely separable state $\tilde{\rho}$ that escapes detection by all existing sufficient conditions but is certified by ours.
\begin{figure}[t]
    \centering
    \def\svgwidth{\linewidth}
    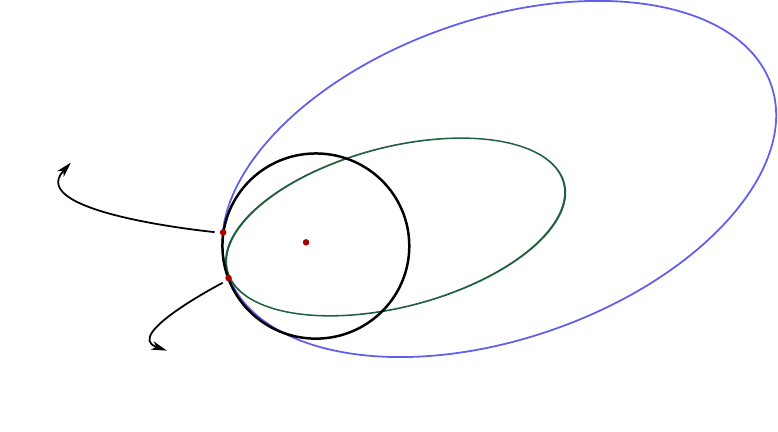
    \caption{Graphical representation of the absolutely separable set (AS), completely absolutely separable set (CAS) and the separable ball introduced in~\cite{PhysRevA.66.062311} (black circle) for unequal local dimensions $d_A$ and $d_B$. The state $(\frac{1}{d_Ad_B-1},...,\frac{1}{d_Ad_B-1},0)$ is on the boundary of both the purity ball and the AS set \cite{hildebrandappt}. The state $\frac{\mathds{1}-\ketbra{\Phi^+}{\Phi^+}^{\Gamma}}{(d_Ad_B)-1}$ ($\Gamma$ is the partial transpose of $B$) is on the boundary of every set. The state $\tilde{\rho}$ lies outside the purity ball. Moreover, it remains undetected by every known criterion for absolute separability in the literature. However, it lies within CAS and thus AS (see Sec.~G of the Supplemental Material).}
    \label{fig:fig_1}
\end{figure}

\emph{Entanglement extraction from purity.---}
A bipartite state $\rho^{AB}$ is called absolutely separable if
\begin{equation}
  U_{AB}\rho^{AB}U_{AB}^\dagger \in \mathcal{S}  
\end{equation}
for every global unitary $U_{AB}$, where $\mathcal{S}$ denotes the set of separable states. Absolute separability has been extensively studied; see, for example, Refs.~\cite{PhysRevA.64.012316,PhysRevA.108.042402,PhysRevA.88.062330,KnillSeparabilitySpectrum,arunachalam2015absoluteseparabilitydeterminedpartial,PhysRevA.103.052431}. Considerable effort has been devoted to characterizing this set.

An equivalent formulation is that absolutely separable states are precisely those bipartite states that remain separable under all unital dynamics. More explicitly, $\rho^{AB}$ is absolutely separable if and only if
\begin{equation}
    \Lambda[\rho^{AB}] \in \mathcal{S}
\end{equation}
for every completely positive trace-preserving (CPTP) map $\Lambda$ satisfying $\Lambda(\mathds{1})=\mathds{1}$.

To see the equivalence between these two definitions of absolute separability, first note that every unitary channel is unital, so invariance under all unital channels immediately implies invariance under all global unitaries. Conversely, if $\rho^{AB}$ is absolutely separable and $\Lambda$ is any unital channel, then $\Lambda[\rho^{AB}]$ is majorized by $\rho^{AB}$. This implies that there exist probabilities $\{p_i\}$ and unitaries $\{U_i\}$ such that~\cite{spekkens,Streltsov_2018}
\begin{equation}
    \Lambda[\rho^{AB}]=\sum_i p_i\, U_i \rho^{AB} U_i^\dagger.
\end{equation}
Since each state $U_i \rho^{AB} U_i^\dagger$ is separable by absolute separability, and since the set $\mathcal{S}$ is convex, it follows that $\Lambda[\rho^{AB}]$ is also separable. Hence the two formulations are equivalent. Therefore, absolutely separable states are precisely those states from which entanglement cannot be generated by any deterministic physical process that preserves the maximally mixed state.

Now consider the absolutely separable two-qubit state
\begin{equation}
\label{asnotcas}
    \rho=\frac{1}{3}\bigl(\ketbra{00}{00}+\ketbra{01}{01}+\ketbra{10}{10}\bigr).
\end{equation}
In Sec.~C of the Supplemental Material, we show that there exists a stochastic quantum operation $\tilde\Lambda$ such that
\begin{equation}\label{generation}
    \tilde\Lambda[\mathds{1}] = \frac{5}{12}\mathds{1}
    \quad \text{and} \quad
    \tilde\Lambda[\rho] = \frac{2}{9}\rho_W,
\end{equation}
where $\rho_W$ is an entangled Werner state. Here, $\tilde\Lambda$ is a completely positive, trace-non-increasing map and therefore corresponds to a valid quantum operation associated with a particular measurement outcome. In this sense, it describes a stochastic process. Moreover, $\tilde\Lambda$ is purity non-generating (unital), in the sense that it does not create purity, even conditionally on the successful outcome. Nevertheless, it probabilistically generates entanglement from the absolutely separable state $\rho$ (see Eq. (\ref{generation})). Note that, any probabilistic unital channel can be completed to a deterministic unital channel by assigning the failure outcome to the preparation of the maximally mixed state. Consequently, every probabilistic unital channel can be embedded in a deterministic unital channel.

This example shows that some absolutely separable states can still be used to generate entanglement with nonzero probability. This naturally motivates the question of characterizing all bipartite states from which entanglement cannot be generated with any nonzero probability. Such states represent the most fundamental limitation on entanglement generation from purity-non-generating processes.

We define a bipartite state $\rho^{AB}$ to be \emph{completely absolutely separable} if
\begin{equation}\label{CAS}
    \frac{\tilde\Lambda[\rho^{AB}]}{\Tr[\tilde\Lambda[\rho^{AB}]]} \in \mathcal{S}
\end{equation}
for every probabilistic unital channel $\tilde\Lambda$, i.e. if no physical
process can generate entanglement from it. We now state our main result.
\begin{theorem}\label{main}
   A bipartite state $\rho^{AB}$ is completely absolutely separable if and only if
\begin{equation}\label{theorem}
    R(\rho^{AB}) \le \frac{d+1}{d-1}.
\end{equation}
Here $R(\rho^{AB}) := \frac{\lambda_{\max}}{\lambda_{\min}}$ is the ratio of the largest and smallest eigenvalues of $\rho^{AB}$, and $d := \min\{d_A,d_B\}$ is the smaller local dimension.
\end{theorem}
This result provides a complete and remarkably simple spectral characterization of completely absolutely separable states. In contrast, a complete characterization of absolutely separable states remains open. Note that, completely absolutely separable states must be full-rank (see Supplemental Material, Sec.~D), so the quantity $R(\rho^{AB})$ is always well defined. 

The proof proceeds in three steps: First, we show that every state satisfying \(R(\rho^{AB}) \le \frac{d+1}{d-1}\) is separable. Second, we show that the bound is tight; see Supplemental Material,
Sec.~E. Third, we use the fact that \(R\) is a complete monotone under probabilistic unital operations (Supplemental Material, Sec.~D), meaning that \(\rho\) can be transformed into \(\sigma\) with nonzero probability if and only if \(R(\rho)\ge R(\sigma)\) \cite{PhysRevLett.128.110505,Regula2022tightconstraints}. These three ingredients together establish the theorem. 

\emph{A novel sufficient condition for separability.---}
A key ingredient in the proof of Theorem~\ref{main} is the following sufficient condition for separability: a bipartite state \(\rho^{AB}\) is separable whenever Eq. (\ref{theorem}) is satisfied. The proof of this statement can be found in Sec.~B of the Supplemental Material and it relies on an important structural property of bipartite entanglement witnesses. Recall that an \emph{entanglement witness} is a Hermitian operator \(W\) such that $\Tr\!\left[\rho^{AB}W\right]\ge 0$ for all separable states \(\rho^{AB}\). This property is also known as block-positivity. At the same time, there exists at least one entangled state \(\rho^{AB}_{\mathrm{ent}}\) such that $\Tr\!\left[\rho^{AB}_{\mathrm{ent}}W\right]<0$, in which case \(W\) is said to \emph{detect} \(\rho^{AB}_{\mathrm{ent}}\). Several spectral properties of entanglement witnesses are known \cite{song2025,Sarbicki_2008,Champagne_2022}. In Sec.~A of the Supplemental Material, we show that any bipartite witness with trace 1, satisfies the following condition
\begin{equation}
  \|W\|_1 \le d:=\min\{d_A,d_B\},  
\end{equation}
where \(\|\cdot\|_1\) denotes the trace norm. This resolves an open question posed in \cite{song2025}. The criterion in Eq.~\eqref{theorem} depends only on the eigenvalue ratio \(R\) and the smaller local dimension. Since it is based solely on the spectrum of the state, it constitutes a spectral sufficient condition for separability. This is especially appealing from a practical standpoint, as spectral information is often more accessible than full state tomography \cite{PhysRevA.64.052311,CIESLINSKI20241,7956181}. To the best of our knowledge, all previously known spectral separability criteria depend on the global dimension and become trivial in the limit of large global dimension 
\cite{PhysRevA.66.062311,PhysRevA.75.062330,Abellanet_Vidal_2025,Lewenstein_2016}. In contrast, our new criterion depends only on the smaller local dimension of the bipartite system. This has important implications for certifying separability from spectral data: in particular, for systems with unequal local dimensions, it detects separable states beyond the reach of previously known spectral criteria.

To illustrate this, consider the largest separable ball criterion of Ref.~\cite{PhysRevA.66.062311}, which guarantees that a state is separable whenever its purity satisfies
\begin{equation}\label{ball}
    \Tr(\rho^2)\le \frac{1}{d_A d_B-1}.
\end{equation}
Now, if one increases one local dimension while keeping the other fixed, the corresponding separable ball shrinks and in the limit where one subsystem dimension becomes arbitrarily large, it collapses to the maximally mixed state. This behavior stands in sharp contrast to our criterion, which remains invariant under an enlargement of one local dimension because it depends only on the smaller local dimension. 

Indeed, in Sec.~G of the Supplemental Material, we construct, for every pair of unequal local dimensions, an explicit absolutely separable state that is not detected either by the purity-ball criterion or by the criterion of Ref.~\cite{Abellanet_Vidal_2025}, but is detected by our condition in Eq.~\eqref{theorem}. By contrast, when the local dimensions are equal, the set of completely absolutely separable states is contained within the purity ball. A proof of this fact is given in Sec.~H of the Supplemental Material.

Note that, since unitary transformations preserve purity, every state satisfying Eq.~(\ref{ball}) is absolutely separable. The converse is false: there exist absolutely separable states that violate Eq.~(\ref{ball}) (see \cite{PhysRevA.66.062311}). We show in Sec.~I of the Supplemental Material that any absolutely separable state must satisfy
\begin{equation}
    \Tr(\rho^2)\leq \frac{2}{d_A d_B}.
\end{equation}
This yields an analytical upper bound on the maximal purity of absolutely separable states and improves upon the previously known analytical bounds of Ref.~\cite{Filippov_2017}.

\emph{Application to multipartite states.---}
Let \(\rho^{A_1,\ldots,A_N}\) be an \(N\)-partite state on subsystems
\(A_1,\ldots,A_N\). Suppose that, for some \(l\in\mathbb{N}\),
\begin{equation}\label{multipartite}
    R(\rho^{A_1,\ldots,A_N})
    := \frac{\lambda_{\max}}{\lambda_{\min}}
    \leq \frac{l+1}{l-1},
\end{equation}
where \(\lambda_{\max}\) and \(\lambda_{\min}\) denote the largest and smallest
eigenvalues of \(\rho^{A_1,\ldots,A_N}\), respectively. Then
\(\rho^{A_1,\ldots,A_N}\) is separable across every bipartition
\[
A_1,\ldots,A_x \,\big|\, A_{x+1},\ldots,A_N
\]
for which at least one side has Hilbert-space dimension at most \(l\). As a simple application, consider the Gibbs state
\[
    \rho_T = \frac{e^{-H/(k_B T)}}{\operatorname{Tr}(e^{-H/(k_B T)})}
\]
of a multipartite Hamiltonian \(H\) with eigenvalues $\{E_i\}_i$. In this
setting, one may ask for temperatures above which entanglement cannot persist.
By Eq.~\eqref{multipartite}, the Gibbs state is separable across any bipartition whose smaller subsystem has dimension at most ($l$), provided that
\begin{equation} T \geq \frac{2\|H\|_\infty} {k_B\,\ln\!\left(\frac{l+1}{l-1}\right)} . \end{equation}
Here, \(\|H\|_\infty:=\max|E_i| \) denotes the Schatten $\infty$-norm of \(H\). We used the bound
\(E_{\max}-E_{\min}\leq 2\|H\|_\infty\).

\emph{Relation between absolute separability and complete absolute separability.---} Recall that a bipartite state \(\rho^{AB}\) is called \emph{absolutely separable} iff $U\rho^{AB}U^\dagger$ is separable for every unitary \(U\) acting on \(AB\). Closely related is the notion of an \emph{absolutely PPT} (APPT) state, namely a state whose partial transpose remains positive under all global unitaries. A longstanding open problem asks whether the sets of absolutely separable and absolutely PPT states coincide for arbitrary local dimensions \cite{KnillSeparabilitySpectrum,arunachalam2015absoluteseparabilitydeterminedpartial}. This question has been answered affirmatively when one of the local dimensions is equal to two \cite{PhysRevA.88.062330} but remains open otherwise.

An important structural property of the set of absolutely separable states (which includes the set of completely absolutely separable states) is that it is not closed under tensor products. Every non maximally mixed bipartite state \(\rho^{AB}\) fails to remain absolutely separable under sufficiently many tensor powers. Precisely, there exists a finite \(n\in\mathbb{N}\) such that \((\rho^{AB})^{\otimes n}\) can be  entangled  by a global unitary and is therefore not absolutely separable. In Sec.~J of the Supplemental Material, we show that it is enough to take
 \begin{equation}
        n >
        \frac{
        \ln\!\left(R(\rho_{AB})+2\sqrt{R(\rho_{AB})}\right)
        }{
        \ln R(\rho_{AB})
        } .
    \end{equation}
with \(R(\rho^{AB})\) the ratio of the largest to the smallest eigenvalue of \(\rho^{AB}\). In fact, there even exist absolutely separable states, which are not closed even under attaching maximally mixed states. To illustrate this, let \(D=d_Ad_B\) denote the total dimension, and recall that any absolutely separable state must have rank at least \(D-1\) \cite{arunachalam2015absoluteseparabilitydeterminedpartial} (see Fig. \ref{fig:fig_1}). Consider an absolutely separable state \(\rho\) whose spectrum consists of \(D-1\) equal nonzero eigenvalues and one zero eigenvalue; such states are known to be absolutely separable \cite{PhysRevA.66.062311}. For any ancillary maximally mixed state of dimension \(\tilde D\), we have
\begin{equation}
    \operatorname{rank}\!\left(\rho\otimes\frac{\mathds{1}}{\tilde D}\right)
=
\operatorname{rank}(\rho)\cdot\tilde D
=
(D-1)\cdot\tilde D<D\tilde D -1.
\end{equation}
Therefore, \(\rho\otimes \frac{\mathds{1}}{\tilde D}\) cannot be absolutely separable, regardless of the bipartition. This naturally leads to the following question: can one characterize the set of bipartite states \(\rho^{AB}\) such that $\rho^{AB}\otimes \frac{\mathds{1}_{B'}}{d_{B'}}$ remains absolutely separable for every \(d_{B'}\in\mathbb{N}\)? In the following theorem, we show that this set coincides exactly with the set of completely absolutely separable states, when $d_A\leq d_B$.

\begin{theorem}
\label{prop_1}
Let \(\rho^{AB}\) be a state acting on \(\mathcal{H}_A\otimes\mathcal{H}_B\). Then the following statements are equivalent:
\begin{enumerate}[label=(\roman*)]
    \item \(\rho^{AB}\otimes \frac{\mathds{1}_{B'}}{d_{B'}}\) is absolutely separable for all \(d_{B'}\in\mathbb{N}\).
    \item \(\rho^{AB}\otimes \frac{\mathds{1}_{B'}}{d_{B'}}\) is absolutely PPT for all \(d_{B'}\in\mathbb{N}\).
    \item \(R(\rho^{AB})\le \frac{d_A+1}{d_A-1}\).
\end{enumerate}
\end{theorem}

The proof is given in Sec.~F of the Supplemental Material. Note that, when \(d_A\le d_B\), this class is precisely the set of completely absolutely separable states (see Theorem (\ref{main})). This motivates the name $\it{completely}$ absolutely separable in analogy to $\it{completely}$ positive maps.

\emph{Conclusion and final remarks.---}
We have investigated the most fundamental limitations on entanglement generation from purity by going beyond deterministic unital dynamics and considering the full class of probabilistic unital channels. In this more general setting, we showed that absolute separability is not the ultimate obstruction to entanglement extraction: certain absolutely separable states can still generate entanglement with nonzero probability. Motivated by this, we introduced the notion of \emph{complete absolute separability} (CAS), capturing precisely those states from which entanglement cannot be generated with any nonzero probability by any probabilistic purity-non-generating process.

Our main result is a complete characterization of CAS states in arbitrary bipartite dimensions. Remarkably, this characterization is purely spectral and depends only on the ratio between the largest and smallest eigenvalues and on the smaller local dimension. Along the way, we obtained a new sufficient condition for separability, also expressed only in terms of the eigenvalue ratio and the smallest local dimension. This criterion is independent of previously known spectral separability tests and is especially powerful for systems with unequal local dimensions, where it detects absolutely separable states beyond the reach of any earlier sufficient conditions. We further showed that the same spectral threshold has multipartite implications, yielding a simple guarantee of separability across broad classes of bipartitions.

Finally, we established a structural connection between CAS and the longstanding separability-from-spectrum problem: CAS states are exactly those that remain absolutely separable after adjoining an arbitrary maximally mixed ancilla locally. We hope that these results help clarify the operational role of purity as a resource for entanglement generation, and that the perspective developed here will prove useful for further progress on absolute separability and multipartite entanglement structure.

\emph{Acknowledgments.---} We thank Anna Sanpera for introducing us to problem of absolute separability. We also thank Eric Chitambar, Michael Epping, Ludovico Lami, Peter Ken Schuhmacher, and Nikolai Wyderka for helpful discussions.

This work has been supported by the German Ministry of Education and Research (Projects QuKuK, BMBF Grant No.~16KIS1618K and QSolid, BMFTR Grant No.~ 13N16163). P.B thanks the DLR Quantum Computing Initiative and the Federal Ministry for Economic Affairs and Climate Action. We also acknowledge financial support by Deutsche Forschungsgemeinschaft
(DFG, German Research Foundation) under Germany’s
Excellence Strategy – Cluster of Excellence Matter and
Light for Quantum Computing (ML4Q) EXC 2004/1 –
390534769.
\bibliography{literature}
\clearpage

\section*{Supplemental Material}

\subsection{Upper bound on the 1-norm of entanglement witnesses}

\begin{lemma}\label{witness_proof}
Let $W$ be a hermitian block-positive operator on
$\mathcal H_{A}\otimes\mathcal H_{B}$ with $\operatorname{Tr} W = 1$.
Then
\begin{equation}
\|W\|_{1}\leq\ d.
\end{equation}
where, $d=\min\{d_A,d_B\}$.
\end{lemma}

\begin{proof}
Without loss of generality, let us assume $d_A\leq d_B$. Rewrite $W$ via its block decomposition by choosing a basis $\{\ket{i}\}_i$ of $\mathcal{H}_A$ 
\begin{equation}
W=\sum_{i,j=1}^{d_A}\ket{i}\bra{j}\otimes \widetilde{W}_{ij}, \qquad t_i:=\Tr \widetilde{W}_{ii}\geq 0, \qquad \sum_{i=1}^{d_A} t_i=1,
\end{equation}
Note that, the condition $\widetilde{W}_{ii}\geq 0$ follows from block-positivity of $W$. Fix $i\neq j$ and a unit vector $\ket{\psi}\in\mathcal{H}_B$. Block positivity gives
\begin{equation}
\begin{aligned}
0&\leq \Bigl((\alpha^{*}\bra{i}+\beta^{*}\bra{j})\otimes\bra{\psi}\Bigr) W \Bigl((\alpha \ket{i}+\beta\ket{j})\otimes\ket{\psi}\Bigr)\\&=|\alpha|^{2} a
+|\beta|^{2} f
+\alpha^{*}\beta\, c
+\beta^{*}\alpha\, c^{*}\\
&=(\alpha^*\,\, \beta^*)\begin{pmatrix} a &c\\ c^* &f\end{pmatrix}\begin{pmatrix}
\alpha\\ \beta
\end{pmatrix} \,\,\forall \,\,\alpha,\beta \in \mathbb{C},
\end{aligned}
\end{equation}
where $a=\braket{\psi|\widetilde{W}_{ii}|\psi}$, $f=\braket{\psi|\widetilde{W}_{jj}|\psi}$, $c=\braket{\psi|\widetilde{W}_{ij}|\psi}$. Note that, since \(W\) is Hermitian, we have
\[
\widetilde{W}_{ij}=\widetilde{W}_{ji}^{\dagger}.
\]
Hence, if \(\braket{\psi|\widetilde{W}_{ij}|\psi}=c\), then
\[
\braket{\psi|\widetilde{W}_{ji}|\psi}
=
\braket{\psi|\widetilde{W}_{ij}^{\dagger}|\psi}
=
\braket{\psi|\widetilde{W}_{ij}|\psi}^{*}
=
c^{*}.
\] Hence the $2\times2$ matrix
$M=\bigl(\!\begin{smallmatrix}a & c\\ c^{*} & f\end{smallmatrix}\!\bigr)$
is positive-semidefinite. Because $M$ is Hermitian, Sylvester's criterion says that $M$ is positive-semidefinite \emph{iff} its principal minors are non-negative, therefore
\begin{equation}
 \det M=af-|c|^2\geq 0.
\end{equation}
 Hence
\begin{equation}
|c|^2\leq af\implies |c|\leq \sqrt{af}.
\end{equation}
The arithmetic geometric mean inequality gives $\sqrt{af}\leq\frac 12 (a+f)$, so altogether
\begin{equation}
\label{star}
|c|\leq\frac 12 (a+f).
\end{equation}
We can write the 1-norm of any linear operator $X$ as $\|X\|_1=\sup_{U}|\Tr (UX)|$ with $U$ a unitary operator. Choose $X$ to be $\widetilde{W}_{ij}$ and choose an orthonormal basis $\{\ket{\psi_k}\}_k$ of $\mathcal{H}_B$. Without loss of generality we can choose $\{\ket{\psi_k}\}_k$ to be the eigenbasis of $U$. Then,
\begin{equation}
\begin{aligned}
|\Tr(U\widetilde{W}_{ij})|&=\Bigl|\sum_k e^{i\theta_k}\braket{\psi_k|\widetilde{W}_{ij}|\psi_k}\Bigr|\leq \sum_k |\braket{\psi_k|\widetilde{W}_{ij}|\psi_k}|\\
&\stackrel{\eqref{star}}{\le}\frac 12 \sum_k \Bigl(\braket{\psi_k|\widetilde{W}_{ii}|\psi_k}+\braket{\psi_k|\widetilde{W}_{jj}|\psi_k} \Bigr)= \frac 12 (t_i+t_j),
\end{aligned}
\end{equation}
with $\theta_k\in [0,2\pi]$ and $t_j=\Tr \widetilde{W}_{jj}$. We used the fact that $U$ is unitary and thus, $U\ket{\psi_k}=e^{i\theta_k}\ket{\psi_k}$. Taking the supremum over $U$ yields
\begin{equation}
\label{star2}
||\widetilde{W}_{ij}||_1\leq \frac12 (t_i+t_j) \qquad \forall i,j.
\end{equation}
Now using that $||\ket{i}\bra{j}\otimes \widetilde{W}_{ij}||_1=||\widetilde{W}_{ij}||_1$ and the triangular inequality
\begin{equation}
\begin{aligned}
||W||_1&\leq \sum_{i,j}^{d_A}||\ket{i}\bra{j}\otimes \widetilde{W}_{ij}||_1=\sum_{i,j}^{d_A}||\widetilde{W}_{ij}||_1\\
&\stackrel{\eqref{star2}}{\leq}\sum_{i,j}^{d_A}\frac12 (t_i+t_j)=d_A.
\end{aligned}
\end{equation}
The last inequality follows from $\sum_{i=1}^{d_A} t_i=1$. This completes the proof.
\end{proof}

\subsection{Proof of the sufficient condition for separability} 

\begin{lemma}
\label{spec-sep-lemma}
A bipartite state $\rho$ is separable if 
\begin{equation}\label{theorem}
    R(\rho) \le \frac{d+1}{d-1},
\end{equation}
where $d=\min\{d_A,d_B\}$.
\end{lemma}

\begin{proof}
By assumption,
\[
    R(\rho)\le \frac{d+1}{d-1},
\]
so \(\lambda_{\min} (\rho)>0\). Define
\[
    \Delta:=I-\frac{\rho}{\lambda_{\max}(\rho)} .
\]
Then \(\Delta\ge 0\), and
\[
    \|\Delta\|_{\infty}
    =
    1-\frac{\lambda_{\min}(\rho)}{\lambda_{\max}(\rho)}
    =
    1-\frac1{R(\rho)}
    \le
    1-\frac{d-1}{d+1}
    =
    \frac{2}{d+1}.
\]

Let \(W\) be any entanglement witness normalized by \(\Tr W=1\). Write
\[
    W=W_+-W_-,
    \qquad
    W_+,W_-\ge 0,
    \qquad
    W_+W_-=0.
\]
Set
\[
    a^+:=\Tr W_+,
    \qquad
    a^-:=\Tr W_-.
\]
Since \(\Tr W=1\),
\[
    a^+-a^-=1.
\]
Moreover, by Lemma \ref{witness_proof},
\[
    \|W\|_1=a^++a^-\le d.
\]
Hence
\[
    a^+\le \frac{d+1}{2}.
\]

Now
\[
    \Tr(\rho W)
    =
    \lambda_{\max}\Tr((I-\Delta)W)
    =
    \lambda_{\max}\bigl(1-\Tr(\Delta W)\bigr).
\]
Since \(\Delta\ge 0\),
\[
    \Tr(\Delta W)
    =
    \Tr(\Delta W_+)-\Tr(\Delta W_-)
    \le
    \Tr(\Delta W_+).
\]
Using \(\Delta\le \|\Delta\|_{\infty}I\), we get
\[
    \Tr(\Delta W_+)
    \le
    \|\Delta\|_{\infty}\Tr W_+
    \le
    \frac{2}{d+1}\cdot \frac{d+1}{2}
    =
    1.
\]
Therefore
\[
    \Tr(\rho W)
    =
    \lambda_{\max}\bigl(1-\Tr(\Delta W)\bigr)
    \ge 0.
\]

Thus \(\Tr(\rho W)\ge 0\) for every entanglement witness \(W\) with
\(\Tr W=1\). Since every entangled state is detected by some entanglement
witness, \(\rho\) is separable.
\end{proof}

\subsection{Example of an absolutely but not completely absolutely separable state}

Consider the two-qubit state
\[
    \rho=\frac13\Bigl(
        \ketbra{00}{00}
        +\ketbra{01}{01}
        +\ketbra{10}{10}
    \Bigr)
    =
    \operatorname{diag}\left(\frac13,\frac13,\frac13,0\right),
\]
which is absolutely separable in \(2\otimes 2\). Let
\[
    \rho_W
    =
    \frac12\ketbra{\Psi^-}{\Psi^-}
    +
    \frac{\mathds{1}}{8},
    \qquad
    \ket{\Psi^-}
    =
    \frac{1}{\sqrt2}(\ket{01}-\ket{10}).
\]
This Werner state is full rank and entangled. Set
\[
    \lambda=\frac52,
    \qquad
    \sigma=\lambda\frac{\mathds{1}}{4}-\rho_W .
\]
Note that, \(\sigma\ge 0\) and
\[
    \Tr\sigma=\lambda-1=\frac32.
\]
Thus the normalized state associated with \(\sigma\) is
\[
    \widehat{\sigma}:=\frac{\sigma}{\Tr\sigma}
    =
    \frac23\sigma .
\]

Now define the sub-POVM effects
\[
    E_\sigma
    =
    \ketbra{11}{11}
    =
    \operatorname{diag}(0,0,0,1),
\]
and
\[
    E_W
    =
    \frac29
    \Bigl(
        \ketbra{00}{00}
        +
        \ketbra{01}{01}
        +
        \ketbra{10}{10}
    \Bigr)
    =
    \operatorname{diag}\left(\frac29,\frac29,\frac29,0\right).
\]
They satisfy
\[
    E_\sigma+E_W
    =
    \operatorname{diag}\left(\frac29,\frac29,\frac29,1\right)
    \le \mathds{1}.
\]
Equivalently, the missing failure effect is
\[
    E_{\mathrm{fail}}
    =
    I_4-E_\sigma-E_W
    =
    \operatorname{diag}\left(\frac79,\frac79,\frac79,0\right)
    \ge 0.
\]
Hence \(\{E_\sigma,E_W,E_{\mathrm{fail}}\}\) is a valid POVM.

Define the stochastic measure-and-prepare map
\[
    \Lambda[X]
    =
    \Tr(E_\sigma X)\,\widehat{\sigma}
    +
    \Tr(E_W X)\,\rho_W .
\]
The map is completely positive because it is measure-and-prepare, and it is
trace-non-increasing because
\[
    E_\sigma+E_W\le \mathds{1}.
\]

We now check stochastic unitality. Note that,
\[
    \Tr(E_\sigma)=1,
    \qquad
    \Tr(E_W)=\frac{2}{3}.
\]
Therefore
\[
\begin{aligned}
    \Lambda[\mathds{1}]
    &=
    \widehat{\sigma}
    +
    \frac23\,\rho_W =\frac{5}{12}\mathds{1}
\end{aligned}
\]
Hence \(\Lambda\) is stochastic unital. Finally, evaluate \(\Lambda\) on \(\rho\). Since \(\rho\) has no support on
\(\ket{11}\),
\[
    \Tr(E_\sigma\rho)=0.
\]
On the other hand,
\[
    \Tr(E_W\rho)
    =
    \frac29
    \left(
        \frac13+\frac13+\frac13
    \right)
    =
    \frac29.
\]
Thus
\[
    \Lambda[\rho]
    =
    \frac29\,\rho_W .
\]
Thus the post-selected output is
\[
    \frac{\Lambda[\rho]}{\Tr\Lambda[\rho]}=\rho_W,
\]
which is entangled, and the success probability is
\[
    \Tr\Lambda[\rho]=\frac{2}{9}.
\]

\subsection{Probabilistic unital transformations}

We now characterize state transformations under stochastic unital maps.
Throughout this section, the input and output systems have the same dimension
\(D\). A stochastic unital map is a completely positive trace-nonincreasing map
\(\Lambda\) such that
\begin{equation}
    \Lambda(\mathds{1}) = q\mathds{1}
\end{equation}
for some \(q>0\).

We first prove a simple necessary condition. Let \(\rho\) be full rank. Then
\begin{equation}
    \lambda_{\min}(\rho)\mathds{1}
    \leq \rho \leq
    \lambda_{\max}(\rho)\mathds{1}.
\end{equation}
Applying a stochastic unital map \(\Lambda\), we obtain
\begin{equation}
    q\lambda_{\min}(\rho)\mathds{1}
    \leq \Lambda(\rho) \leq
    q\lambda_{\max}(\rho)\mathds{1}.
\end{equation}
Assume that the transformation succeeds with nonzero probability and define
\begin{equation}
    \sigma =
    \frac{\Lambda(\rho)}{\Tr[\Lambda(\rho)]}.
\end{equation}
It follows that
\begin{equation}
    \lambda_{\min}(\sigma)
    \geq
    \frac{q\lambda_{\min}(\rho)}{\Tr[\Lambda(\rho)]},
    \qquad
    \lambda_{\max}(\sigma)
    \leq
    \frac{q\lambda_{\max}(\rho)}{\Tr[\Lambda(\rho)]}.
\end{equation}
Hence
\begin{equation}
    \frac{\lambda_{\max}(\sigma)}{\lambda_{\min}(\sigma)}
    \leq
    \frac{\lambda_{\max}(\rho)}{\lambda_{\min}(\rho)} .
    \label{eq:ratio-monotone}
\end{equation}
In particular, a full-rank state cannot be transformed with nonzero probability
into a non-full-rank state.

We now prove sufficiency. We show that if either \(\rho\) is singular, or both
\(\rho\) and \(\sigma\) are full rank and
\[
    \frac{\lambda_{\max}(\rho)}{\lambda_{\min}(\rho)}
    \ge
    \frac{\lambda_{\max}(\sigma)}{\lambda_{\min}(\sigma)},
\]
then there exists a stochastic unital map transforming \(\rho\) into
\(\sigma\) with nonzero probability.

If \(\sigma=\mathds{1}/D\), the claim is immediate, since the channel
\[
    \mathcal E(X)=\Tr(X)\frac{\mathds{1}}{D}
\]
is unital and maps every state to \(\mathds{1}/D\). Hence assume
\(\sigma\neq \mathds{1}/D\).

First suppose that \(\sigma\) is full rank. Define
\[
    \alpha=D\lambda_{\max}(\sigma),
    \qquad
    \beta=\frac{1}{D\lambda_{\min}(\sigma)} .
\]
Then \(\alpha,\beta>1\), and
\[
    \alpha\beta
    =
    \frac{\lambda_{\max}(\sigma)}{\lambda_{\min}(\sigma)} .
\]
Set
\[
    \phi_1=
    \frac{\sigma-\beta^{-1}\frac{\mathds{1}}{D}}{1-\beta^{-1}},
    \qquad
    \phi_2=
    \frac{\alpha\frac{\mathds{1}}{D}-\sigma}{\alpha-1}.
\]
Since
\[
    \beta^{-1}\frac{\mathds{1}}{D}
    \le \sigma
    \le
    \alpha\frac{\mathds{1}}{D},
\]
both \(\phi_1\) and \(\phi_2\) are states. Moreover,
\[
    (1-\beta^{-1})\phi_1
    +
    (\alpha-1)\phi_2
    =
    \left(\alpha-\beta^{-1}\right)\frac{\mathds{1}}{D}.
\]

Choose unit vectors \(\ket{x},\ket{y}\) such that
\[
    \frac{\bra{x}\rho\ket{x}}
         {\bra{y}\rho\ket{y}}
    =
    \alpha\beta .
\]
This is possible by the assumed spectral-ratio condition if \(\rho\) is full
rank. If \(\rho\) is singular, choose \(\ket{x}\) as a maximal-eigenvalue
eigenvector and interpolate with a vector in \(\ker\rho\).

Now set
\[
    k=\frac{\alpha-1}{1-\beta^{-1}}>0,
    \qquad
    0<c\le \frac{1}{1+k},
\]
and define
\[
    M_1=c\ketbra{x}{x},
    \qquad
    M_2=ck\ketbra{y}{y}.
\]
Then
\[
    M_1+M_2\le c(1+k)\mathds{1}\le \mathds{1},
\]
so \(\{M_1,M_2\}\) is a valid sub-POVM. Define
\[
    \mathcal E_{\rm s}(X)
    =
    \Tr(M_1X)\phi_1+\Tr(M_2X)\phi_2 .
\]
This map is completely positive and trace nonincreasing. Furthermore,
\[
\begin{aligned}
    \mathcal E_{\rm s}(\mathds{1})
    &=
    c\phi_1+ck\phi_2  \\
    &=
    \frac{c}{1-\beta^{-1}}
    \left[
        (1-\beta^{-1})\phi_1
        +
        (\alpha-1)\phi_2
    \right]  \\
    &=
    \frac{c(\alpha-\beta^{-1})}{1-\beta^{-1}}
    \frac{\mathds{1}}{D}.
\end{aligned}
\]
Hence \(\mathcal E_{\rm s}\) is stochastic unital.

It remains to check the output on \(\rho\). Let
\[
    a=\Tr(M_1\rho),
    \qquad
    b=\Tr(M_2\rho).
\]
Then
\[
    \frac{b}{a}
    =
    k\frac{\bra{y}\rho\ket{y}}{\bra{x}\rho\ket{x}}
    =
    \frac{k}{\alpha\beta}
    =
    \frac{\alpha-1}{\alpha(\beta-1)}
    =:r.
\]
A direct substitution gives
\[
    \phi_1+r\phi_2=(1+r)\sigma .
\]
Therefore
\[
    \mathcal E_{\rm s}(\rho)
    =
    a\phi_1+b\phi_2
    =
    a(\phi_1+r\phi_2)
    =
    (a+b)\sigma .
\]
Since \(a+b>0\), the normalized post-selected output is \(\sigma\).

It remains only to handle the case where \(\rho\) is singular and \(\sigma\)
is arbitrary. Choose
\[
    \alpha>D\lambda_{\max}(\sigma)
\]
and define
\[
    \phi_1=\sigma,
    \qquad
    \phi_2=
    \frac{\alpha\frac{\mathds{1}}{D}-\sigma}{\alpha-1}.
\]
Then \(\phi_2\) is a state. Let \(\ket{x}\) be a maximal-eigenvalue eigenvector
of \(\rho\), and choose \(\ket{y}\in\ker\rho\). Set
\[
    k=\alpha-1,
    \qquad
    0<c\le \frac1{1+k},
\]
and define
\[
    M_1=c\ketbra{x}{x},
    \qquad
    M_2=ck\ketbra{y}{y}.
\]
Again \(M_1+M_2\le \mathds{1}\), so the map
\[
    \mathcal E_{\rm s}(X)
    =
    \Tr(M_1X)\phi_1+\Tr(M_2X)\phi_2
\]
is completely positive and trace nonincreasing. Moreover,
\[
    \mathcal E_{\rm s}(\mathds{1})
    =
    c\sigma+c(\alpha-1)\phi_2
    =
    c\alpha\frac{\mathds{1}}{D},
\]
so it is stochastic unital. Since \(\ket{y}\in\ker\rho\),
\[
    \mathcal E_{\rm s}(\rho)
    =
    c\lambda_{\max}(\rho)\sigma .
\]
Thus the transformation succeeds with nonzero probability and the normalized
output is \(\sigma\).

This proves the sufficiency of the stated condition.

\subsection{Proof of theorem 1}

\begin{theorem}\label{main}
A bipartite state \(\rho^{AB}\) is completely absolutely separable if and only if
\[
   R(\rho^{AB}) \le \frac{d+1}{d-1},
\]
where
\[
    R(\rho):=\frac{\lambda_{\max}(\rho)}{\lambda_{\min}(\rho)},
    \qquad
    d=\min\{d_A,d_B\}.
\]
\end{theorem}

\begin{proof}
Suppose first that
\[
    R(\rho)\le \frac{d+1}{d-1}.
\]
Then \(\rho\) is full rank. Let \(\Lambda\) be any stochastic unital operation
with \(\Lambda(\rho)\neq 0\), and define
\[
    \sigma=\frac{\Lambda(\rho)}{\Tr[\Lambda(\rho)]}.
\]
By monotonicity of the spectral ratio under stochastic unital maps,
\[
    R(\sigma)\le R(\rho)\le \frac{d+1}{d-1}.
\]
Hence, by Lemma~\ref{spec-sep-lemma}, \(\sigma\) is separable. Thus no
stochastic unital operation can produce entanglement from \(\rho\), so \(\rho\)
is completely absolutely separable.

Conversely, suppose
\[
    R(\rho)> \frac{d+1}{d-1}.
\]
We show that \(\rho\) can be mapped probabilistically to an entangled state.

Let
\[
    \ket{\Phi^+}
    =
    \frac1{\sqrt d}\sum_{i=1}^d \ket{i}_A\otimes\ket{i}_B,
    \qquad
    W=\ketbra{\Phi^+}{\Phi^+}^\Gamma .
\]
Here \(\Gamma\) denotes partial transposition with respect to \(B\). For \(0\le t<d\), define
\[
    \omega_t
    =
    \frac{\mathds{1}-tW}{D-t},
    \qquad
    D=d_Ad_B .
\]
Since \(W\) has spectrum contained in \([-1/d,1/d]\), the operator
\(\omega_t\) is a state. Its spectral ratio is
\[
    R(\omega_t)
    =
    \frac{1+t/d}{1-t/d}.
\]
Moreover,
\[
    \omega_t^\Gamma
    =
    \frac{\mathds{1}-t\ketbra{\Phi^+}{\Phi^+}}{D-t}.
\]
Thus, whenever \(t>1\), the partial transpose has eigenvalue
\[
    \frac{1-t}{D-t}<0,
\]
so \(\omega_t\) is entangled. Now choose
\[
    t=d\left(\frac{R(\rho)-1}{R(\rho)+1}\right).
\]
The assumption \(R(\rho)>(d+1)/(d-1)\) implies \(t>1\), and clearly \(t<d\).
Therefore \(\omega_t\) is entangled and satisfies
\[
    R(\omega_t)=R(\rho).
\]

If \(\rho\) is full rank, the stochastic-unital transformation criterion gives
a stochastic unital map sending \(\rho\) to \(\omega_t\) with nonzero
probability, since \(R(\rho)=R(\omega_t)\). If \(\rho\) is not full rank, the
singular-input case of the same criterion allows a stochastic unital
transformation from \(\rho\) to any target state, in particular to
\(\omega_t\).

Hence, whenever
\[
    R(\rho)> \frac{d+1}{d-1},
\]
the state \(\rho\) can be probabilistically transformed by a stochastic unital
operation into an entangled state. Therefore \(\rho\) is not completely
absolutely separable. This proves the theorem.
\end{proof}

\subsection{Proof of Theorem 2}
\begin{theorem}
Let \(\rho\) be a quantum state on
\(\mathcal H_A\otimes\mathcal H_B\), with \(d_A\le d_B\). Then the
following statements are equivalent:
\begin{enumerate}[label=(\roman*)]
    \item
    \(\rho\otimes \frac{\mathds{1}_{B'}}{d_{B'}}\) is absolutely separable
    across the bipartition \(A:BB'\) for every \(d_{B'}\in\mathbb N\).

    \item
    \(\rho\otimes \frac{\mathds{1}_{B'}}{d_{B'}}\) is absolutely PPT
    across the bipartition \(A:BB'\) for every \(d_{B'}\in\mathbb N\).

    \item
    \[
        R(\rho)\le \frac{d_A+1}{d_A-1}.
    \]
\end{enumerate}
\end{theorem}

\begin{proof}
We prove the implications
\[
    \text{(ii)}\Rightarrow \text{(iii)}
    \Rightarrow \text{(i)}
    \Rightarrow \text{(ii)} .
\]

\medskip\noindent
\emph{Proof of \((ii)\Rightarrow(iii)\).}
Assume that \(\rho\otimes \frac{\mathds{1}_{B'}}{d_{B'}}\) is absolutely PPT for every
\(d_{B'}\). Thus, for every unitary \(U\) and every vector \(\ket{\psi}\),
\[
    \Tr\left[
        \left(
            U\left(\rho\otimes\frac{\mathds{1}_{B'}}{d_{B'}}\right)U^\dagger
        \right)^\Gamma
        \ketbra{\psi}{\psi}
    \right]\ge 0.
\]
Equivalently,
\[
    \Tr\left[
        \left(\rho\otimes\frac{\mathds{1}_{B'}}{d_{B'}}\right)
        U^\dagger \ketbra{\psi}{\psi}^{\Gamma} U
    \right]\ge 0 .
\]
Here \(\Gamma\) denotes partial transposition with respect to \(BB'\).

Since condition (ii) holds for every \(d_{B'}\), choose
\begin{equation}\label{range}
    d_{B'}\ge \frac{d_A(d_A+1)}{2}.
\end{equation}
Write the Schmidt decomposition of \(\ket{\psi}\) across \(A:BB'\) as
\[
    \ket{\psi}
    =
    \sum_{i=1}^{d_A}\sqrt{p_i}\,
    \ket{i}_A\otimes\ket{i}_{BB'},
    \qquad
    p_i\ge 0,
    \qquad
    \sum_i p_i=1.
\]
The spectrum of \(\ketbra{\psi}{\psi}^{\Gamma}\) consists of
\[
    p_i,
    \qquad
    \pm\sqrt{p_i p_j}\quad (i<j),
\]
together with possible additional zero eigenvalues. Hence it has
\[
    \frac{d_A(d_A-1)}{2}
\]
negative eigenvalues and
\[
    \frac{d_A(d_A+1)}{2}
\]
positive eigenvalues. The eigenvalues of
\[
    \rho\otimes\frac{\mathds{1}_{B'}}{d_{B'}}
\]
are the eigenvalues of \(\rho\) divided by \(d_{B'}\), each repeated
\(d_{B'}\) times. By the choice of \(d_{B'}\) (see Eq. (\ref{range})), there are enough copies of
\(\lambda_{\max}(\rho)/d_{B'}\) to pair with all negative eigenvalues of
\(\ketbra{\psi}{\psi}^{\Gamma}\), and enough copies of
\(\lambda_{\min}(\rho)/d_{B'}\) to pair with all positive eigenvalues.

Using the rearrangement inequality for Hermitian matrices, we therefore obtain
\[
\begin{aligned}
0
&\le
\min_U
\Tr\left[
    \left(\rho\otimes\frac{\mathds{1}_{B'}}{d_{B'}}\right)
    U^\dagger \ketbra{\psi}{\psi}^{\Gamma} U
\right] \\
&=
-\frac{\lambda_{\max}(\rho)}{d_{B'}}
\sum_{i<j}\sqrt{p_i p_j}
+
\frac{\lambda_{\min}(\rho)}{d_{B'}}
\left(
    \sum_i p_i
    +
    \sum_{i<j}\sqrt{p_i p_j}
\right).
\end{aligned}
\]
Since \(\sum_i p_i=1\), this implies
\[
    \lambda_{\max}(\rho)
    \sum_{i<j}\sqrt{p_i p_j}
    \le
    \lambda_{\min}(\rho)
    \left(
        1+\sum_{i<j}\sqrt{p_i p_j}
    \right).
\]
Thus, whenever \(\sum_{i<j}\sqrt{p_i p_j}>0\),
\[
    R(\rho)
    \le
    1+
    \frac{1}{\sum_{i<j}\sqrt{p_i p_j}} .
\]

Now choose \(\ket{\psi}\) maximally entangled across a \(d_A\)-dimensional
subspace, i.e.
\[
    p_i=\frac1{d_A}
    \qquad
    \text{for all } i=1,\dots,d_A .
\]
Then
\[
    \sum_{i<j}\sqrt{p_i p_j}
    =
    \frac{d_A-1}{2}.
\]
Therefore
\[
    R(\rho)
    \le
    1+\frac{2}{d_A-1}
    =
    \frac{d_A+1}{d_A-1}.
\]
This proves (iii).

\medskip\noindent
\emph{Proof of \((iii)\Rightarrow(i)\).}
Assume
\[
    R(\rho)\le \frac{d_A+1}{d_A-1}.
\]
For every \(d_{B'}\),
\[
    R\left(
        \rho\otimes\frac{\mathds{1}_{B'}}{d_{B'}}
    \right)
    =
    R(\rho).
\]
Moreover, for the bipartition \(A:BB'\), the smaller local dimension is
\(d_A\). Hence for every global unitary \(U\),
\[
    R\left(
        U\left(\rho\otimes\frac{\mathds{1}_{B'}}{d_{B'}}\right)U^\dagger
    \right)
    =
    R(\rho)
    \le
    \frac{d_A+1}{d_A-1}.
\]
By Lemma~\ref{spec-sep-lemma}, every state in the unitary orbit of
\[
    \rho\otimes\frac{\mathds{1}_{B'}}{d_{B'}}
\]
is separable across \(A:BB'\). Therefore
\[
    \rho\otimes\frac{\mathds{1}_{B'}}{d_{B'}}
\]
is absolutely separable for every \(d_{B'}\). This proves (i).

\medskip\noindent
\emph{Proof of \((i)\Rightarrow(ii)\).}
This is immediate, since every separable state is PPT. Therefore absolute
separability implies absolute PPT.

Combining the three implications proves the equivalence.
\end{proof}

\subsection{CAS states which cannot be detected by other existing spectral separability conditions}

The best previously known spectral sufficient condition for absolute
separability that we compare with is the convex hull of the two regions \cite{Abellanet_Vidal_2025}
\begin{align}
    A
    &:=
    \left\{\rho:\rho\succeq \frac{\mathds{1}}{D+2},\ \Tr\rho=1\right\},\label{1st}
    \\
    B
    &:=
    \left\{\rho:\Tr(\rho^2)\le \frac{1}{D-1},\ \Tr\rho=1\right\}. \label{2nd}
\end{align}
We denote this region by
\[
    C:=\operatorname{conv}(A\cup B).
\]
The next lemma shows that the eigenvalue-ratio criterion detects states that
are not contained in this convex hull.

\begin{lemma}
Let \(2\le d_A<d_B\), \(D=d_Ad_B\), and
\[
    R:=\frac{d_A+1}{d_A-1}.
\]
Set
\[
    p:=\lfloor D/2\rfloor,
    \qquad
    q:=\lceil D/2\rceil,
    \qquad
    \ell:=\frac{1}{p+qR}.
\]
Define
\[
    \tilde\rho
    =
    \operatorname{diag}
    \left(
        \underbrace{\ell,\ldots,\ell}_{p},
        \underbrace{R\ell,\ldots,R\ell}_{q}
    \right).
\]
Then
\[
    R(\tilde\rho)=R=\frac{d_A+1}{d_A-1},
\]
but
\[
    \tilde\rho\notin C.
\]
Thus the eigenvalue-ratio criterion is not contained in the previously known
convex-hull criterion.
\end{lemma}

\begin{proof}
The identity \(R(\tilde\rho)=R\) is immediate from the definition.

Let \(P\) be the spectral projection of \(\tilde\rho\) corresponding to the
eigenvalue \(\ell\), and define
\[
    Z:=\frac{1}{p}P-\frac{1}{q}(\mathds{1}-P).
\]
Then
\[
    \Tr Z=0,
    \qquad
    \|Z\|_2^2=\frac1p+\frac1q=\frac{D}{pq}.
\]
Let us define $W$
\[
    W:=
    \frac{\mathds{1}}{D}
    +
    \sqrt{\frac{D-1}{D}}\,
    \frac{Z}{\|Z\|_2}.
\]
We show that \(W\) is nonnegative on \(C\), but negative on \(\tilde\rho\).

First let \(\rho\in B\) (see Eq. (\ref{2nd})). Writing \(X=\rho-\mathds{1}/D\), we have
\[
    \Tr X=0,
    \qquad
    \|X\|_2^2
    =
    \Tr(\rho^2)-\frac1D
    \le
    \frac{1}{D(D-1)}.
\]
Hence, by Cauchy--Schwarz,
\[
    \Tr(W\rho)
    =
    \frac1D
    +
    \sqrt{\frac{D-1}{D}}\,
    \frac{\Tr(ZX)}{\|Z\|_2}
    \ge
    \frac1D
    -
    \sqrt{\frac{D-1}{D}}\,
    \frac{1}{\sqrt{D(D-1)}}
    =
    0.
\]
Thus \(W\) is nonnegative on \(B\).

Next let \(\rho\in A\) (see Eq. (\ref{1st})). Write
\[
    \rho=\frac{\mathds{1}}{D+2}+X,
    \qquad
    X\succeq0,
    \qquad
    \Tr X=\frac{2}{D+2}.
\]
Since
\[
    \lambda_{\min}(W)
    =
    \frac1D
    -
    \frac1D\sqrt{\frac{(D-1)p}{q}},
\]
we get
\[
\begin{aligned}
    \Tr(W\rho)
    &\ge
    \frac{1}{D+2}
    +
    \frac{2}{D+2}\lambda_{\min}(W)  \\
    &=
    \frac{D+2-2\sqrt{\frac{(D-1)p}{q}}}{D(D+2)}
    \ge 0.
\end{aligned}
\]
Indeed, \(p\le q\), so whenever $D>4$
\[
    2\sqrt{\frac{(D-1)p}{q}}\le 2\sqrt{D-1}<D+2.
\]
Therefore \(W\) is nonnegative on \(A\), and hence on
\(C=\operatorname{conv}(A\cup B)\).

It remains to evaluate \(W\) on \(\tilde\rho\). Since
\[
    \Tr(Z\tilde\rho)=\ell-R\ell=-(R-1)\ell,
\]
and
\[
    \ell=\frac{1}{p+qR},
    \qquad
    \|Z\|_2^2=\frac{D}{pq},
\]
we get
\[
    \Tr(W\tilde\rho)
    =
    \frac1D
    -
    \frac{(R-1)}{p+qR}
    \sqrt{\frac{D-1}{D}}\,
    \frac{1}{\|Z\|_2}.
\]
Equivalently,
\[
    \Tr(W\tilde\rho)
    =
    \frac1D
    \left[
        1
        -
        \frac{(R-1)\sqrt{(D-1)pq}}{p+qR}
    \right].
\]
Hence \(\Tr(W\tilde\rho)<0\) iff
\[
    (R-1)\sqrt{(D-1)pq}>p+qR.
    \tag{\(*\)}
\]

For \(R=(d_A+1)/(d_A-1)\), \(D=d_Ad_B\), and
\(p=\lfloor D/2\rfloor\), \(q=\lceil D/2\rceil\), inequality \((*)\) holds
whenever \(2\le d_A<d_B\). Indeed, if \(D\) is even, \((*)\) reduces to
\[
    \sqrt{D-1}>d_A,
\]
which follows from \(D=d_Ad_B\) and \(d_B>d_A\). If \(D\) is odd, \((*)\)
reduces, after squaring, to
\[
    d_Ad_B
    \left[
        d_A^2d_B(d_B-d_A)-d_Ad_B-2d_A-1
    \right]>0,
\]
which is true for all \(d_B>d_A\); the bracket is increasing in \(d_B\) and at
\(d_B=d_A+1\) equals \(d_A^3-3d_A-1>0\).

Therefore \(\Tr(W\tilde\rho)<0\). Since \(W\) is nonnegative on \(C\), it separates \(\tilde\rho\) from \(C\). Hence
\[
    \tilde\rho\notin C.
\]
\end{proof}

\subsection{Maximum purity of CAS states}
    \label{sec:H}
\begin{lemma}
    \label{lem:lemma4}
Let \(\rho\) be a completely absolutely separable state on \(\mathcal H_A\otimes \mathcal H_B\), where
\(\dim\mathcal H_A=d_A\), \(\dim\mathcal H_B=d_B\), and \(d_A\le d_B\). Then
\begin{equation}
\label{lemma4}
    \Tr(\rho^2)\le \left(\frac{d_A}{d_B}\right)\,\frac{1}{d_A^2-1}.
\end{equation}
\end{lemma}

\begin{proof}
Write \(D:=d_A d_B\) for the total dimension and let \(\{\lambda_i\}_{i=1}^D\) be the eigenvalues of \(\rho\).
Denote by \(\lambda_{\max}\) and \(\lambda_{\min}\) the largest and smallest eigenvalue, respectively.
We start from a simple inequality that relates the purity to the extreme eigenvalues.

For every \(i\), the eigenvalue \(\lambda_i\) lies between \(\lambda_{\min}\) and \(\lambda_{\max}\), i.e.,
\begin{equation}
\lambda_{\max}-\lambda_i\ge 0
\qquad\text{and}\qquad
\lambda_i-\lambda_{\min}\ge 0.
\end{equation}
Therefore,
\begin{equation}
(\lambda_{\max}-\lambda_i)(\lambda_i-\lambda_{\min})\ge 0
\qquad\text{for all } i.
\end{equation}
Summing over \(i\) and expanding the product gives
\begin{equation}
    \label{eq:eq80}
\begin{aligned}
0
&\le \sum_{i=1}^D (\lambda_{\max}-\lambda_i)(\lambda_i-\lambda_{\min})\\
&= \sum_{i=1}^D\bigl(\lambda_{\max}\lambda_i-\lambda_{\max}\lambda_{\min}-\lambda_i^2+\lambda_{\min}\lambda_i\bigr)\\
&= \lambda_{\max}\sum_{i=1}^D \lambda_i \;-\; D\lambda_{\max}\lambda_{\min}
\;-\;\sum_{i=1}^D \lambda_i^2 \;+\; \lambda_{\min}\sum_{i=1}^D \lambda_i.
\end{aligned}
\end{equation}
Using \(\sum_i \lambda_i=\Tr(\rho)=1\) and \(\sum_i\lambda_i^2=\Tr(\rho^2)\), we obtain
\begin{equation}
0\le (\lambda_{\max}+\lambda_{\min}) - D\lambda_{\max}\lambda_{\min} - \Tr(\rho^2),
\end{equation}
or equivalently,
\begin{equation}\label{eq:purity-extremes}
\Tr(\rho^2)\le \lambda_{\max}+\lambda_{\min}-D\lambda_{\max}\lambda_{\min}=\lambda_{\max} (1-D\lambda_{\min})+\lambda_{\min}.
\end{equation}

Next we use complete absolute separability. By assumption
\(\lambda_{\max}\le \left(\frac{d_A+1}{d_A-1}\right)\lambda_{\min}\). Plugging this into
\eqref{eq:purity-extremes} yields a bound that depends only on \(\lambda_{\min}\):
\begin{equation}
\label{eq:puritybound1}
    \Tr(\rho^2)\le \left(\frac{d_A+1}{d_A-1}+1\right)\lambda_{\min}-D\frac{d_A+1}{d_A-1}\,\lambda_{\min}^2.
\end{equation}

The right hand side of \eqref{eq:puritybound1}
is a concave quadratic function of $\lambda_{\min}$, hence it attains its maximum at the unique critical point
\begin{equation}
\label{mineigencas}
\lambda_{\min}=\frac{1}{d_B(d_A+1)}.
\end{equation}
Therefore,
\begin{equation}
\Tr[\rho^2]\le \left(\frac{d_A}{d_B}\right)\,\frac{1}{d_A^2-1}
\end{equation}
which is the desired bound.
\end{proof}
Note that, for equal local dimensions (i.e. $d_A=d_B$), Eq.~\eqref{lemma4} coincides with the purity ball criterion $\left(\Tr(\rho^2)\le \frac{1}{d_Ad_B-1}\right)$. This means that, for equal local dimensions, CAS is a subset of the purity ball.

\subsection{Maximum purity of AS states}
For two qubits, \(d_A=d_B=2\), the maximal purity achievable by an absolutely
separable state is known exactly and is equal to \(3/8\)
\cite{dung}. The next lemma gives an upper
bound on the corresponding maximal purity in total dimension
\(D=d_A d_B>4\).
\begin{lemma}
    There does not exist any absolutely separable state with purity more than $2/D$, where $D$ is the total dimension of the system.
\end{lemma}
\begin{proof}
 We now combine the bound~\eqref{eq:purity-extremes} with the APPT necessary condition. By the arithmetic--geometric mean inequality,
\begin{equation}
2\sqrt{\lambda_D\lambda_{D-2}}\le \lambda_D+\lambda_{D-2}.
\end{equation}
Hence the APPT necessary condition
\begin{equation}
\lambda_{\max}\le \lambda_{D-1}+2\sqrt{\lambda_D\lambda_{D-2}}
\end{equation}
implies the simpler bound
\begin{equation}\label{simple1}
\lambda_{\max}\le \lambda_{D-2}+\lambda_{D-1}+\lambda_D.
\end{equation}

Since \(\lambda_D=\lambda_{\min}\), and since \(\lambda_{D-2},\lambda_{D-1}\) are the two smallest eigenvalues among \(\lambda_2,\dots,\lambda_{D-1}\), we have (for $D>4$)
\begin{equation}\label{simple2}
\lambda_{D-2}+\lambda_{D-1}
\le \frac{2}{D-2}\sum_{i=2}^{D-1}\lambda_i
= \frac{2}{D-2}(1-\lambda_{\max}-\lambda_{\min}).
\end{equation}
Therefore, combining Eq. (\ref{simple1}) and Eq. (\ref{simple2}), we get
\begin{equation}
\lambda_{\max}
\le \frac{2}{D-2}(1-\lambda_{\max}-\lambda_{\min})+\lambda_{\min},
\end{equation}
which simplifies to
\begin{equation}
D\lambda_{\max}\le 2+(D-4)\lambda_{\min}.
\end{equation}

Substituting this into~\eqref{eq:purity-extremes}, we obtain
\begin{equation}
\begin{aligned}
\Tr(\rho^2)
&\le \lambda_{\max}+\lambda_{\min}-D\lambda_{\max}\lambda_{\min} \\
&=\lambda_{\max} (1-D\lambda_{\min})+\lambda_{\min}\\
&\le \frac{2+(D-4)\lambda_{\min}}{D}
+\lambda_{\min}
-\lambda_{\min}\bigl(2+(D-4)\lambda_{\min}\bigr) \\
&= \frac{2}{D}-\frac{4}{D}\lambda_{\min}-(D-4)\lambda_{\min}^2 \\
&\le \frac{2}{D}.
\end{aligned}
\end{equation}
Thus,
\begin{equation}
\Tr(\rho^2)\le \frac{2}{D}.
\end{equation}
\end{proof}
In higher
dimensions ($D>4$), only upper bounds (Proposition~1 of
Ref.~\cite{Filippov_2017}) were previously known. We now compare it with the bound obtained above and show that our bound of $2/D$ is strictly stronger. In our notation, Proposition~1 of Ref.~\cite{Filippov_2017} states that, if
\[
    \frac{1}{k}\le \Tr(\rho^2)\le \frac{1}{k-1},
\]
then every absolutely separable state satisfies
\[
    1+\sqrt{\frac{k\Tr(\rho^2)-1}{k-1}}
    \le
    3k\sqrt{\frac{\Tr(\rho^2)}{D+8}}.
\]
Substituting
\(\Tr(\rho^2)=2/D\), the condition is satisfied strictly. Indeed, for even
\(D\), choosing \(k=D/2\), the condition becomes
\[
    1
    <
    3\sqrt{\frac{D}{2(D+8)}} ,
\]
which holds for all even \(D\ge4\). For odd \(D\), choosing
\(k=(D+1)/2\), the condition becomes
\[
    1+\sqrt{\frac{2}{D(D-1)}}
    <
    \frac{3(D+1)}{2}\sqrt{\frac{2}{D(D+8)}} ,
\]
which holds for all odd \(D\ge5\). Thus the condition of
Proposition~1 in Ref.~\cite{Filippov_2017} is satisfied strictly at
\(\Tr(\rho^2)=2/D\). By continuity, it remains satisfied for all values of
\(\Tr(\rho^2)\) in some interval above \(2/D\). Hence Proposition~1 of
Ref.~\cite{Filippov_2017} allows purities strictly larger than \(2/D\), whereas
the lemma above excludes all such purities. Therefore, the bound obtained here
is strictly stronger than the bound of
Ref.~\cite{Filippov_2017}.

\subsection{Proof that AS is not closed under k-copy tensor product}

\begin{lemma}
\label{lemma:tpclosure}
Let $\rho_{AB}$ be a full-rank absolutely separable state distinct from the maximally mixed state. Then for every natural number
\begin{equation}
    n >
    \frac{
    \ln\!\left(R(\rho_{AB})+2\sqrt{R(\rho_{AB})}\right)
    }{
    \ln R(\rho_{AB})
    },
\end{equation}
we have $\rho_{AB}^{\otimes n}\notin AS$, where
$R(\rho_{AB}):=\frac{\lambda_1}{\lambda_D}$ and
$\lambda_1\geq\lambda_2\geq\cdots\geq\lambda_D>0$ are the eigenvalues of
$\rho_{AB}$.
\end{lemma}   

\begin{proof}
 Recall that absolutely separable states are absolutely PPT. Therefore, it is
sufficient to show that $\rho_{AB}^{\otimes n}$ violates a necessary condition
for absolute PPT. An absolutely PPT state with decreasingly ordered spectrum
$\{\lambda_1,\ldots,\lambda_D\}$ satisfies the spectral inequality
\cite{hildebrandappt}
    \begin{equation}
        \label{in:abppt}
     \lambda_1\leq \lambda_{D-1}+2\sqrt{\lambda_D\lambda_{D-2}} .
    \end{equation}

For the spectrum of $n\geq 2$ copies of the state $\rho_{AB}$, the largest
eigenvalue is $\lambda_1^n$, while the three smallest eigenvalues are
    \begin{equation}
        \lambda_D^n,\qquad
        \lambda_D^{n-1}\lambda_{D-1},\qquad
        \lambda_D^{n-1}\lambda_{D-1}.
    \end{equation}
Hence, in order for $\rho_{AB}^{\otimes n}$ to be absolutely PPT, its
eigenvalues have to satisfy
    \begin{align}
        \lambda_1^n
        &\leq
        \lambda_D^{n-1}\lambda_{D-1}
        +2\sqrt{\lambda_D^n\lambda_D^{n-1}\lambda_{D-1}} \nonumber\\
        &=\lambda_D^n
        \left(
        \frac{\lambda_{D-1}}{\lambda_D}
        +2\sqrt{\frac{\lambda_{D-1}}{\lambda_D}}
        \right).
        \label{eq:abpptproof}
    \end{align}
Dividing by $\lambda_D^n>0$, we obtain the necessary condition
    \begin{equation}
        R(\rho_{AB})^n
        \leq
        \frac{\lambda_{D-1}}{\lambda_D}
        +2\sqrt{\frac{\lambda_{D-1}}{\lambda_D}} .
    \end{equation}
Since
    \begin{equation}
        \frac{\lambda_{D-1}}{\lambda_D}
        \leq
        \frac{\lambda_1}{\lambda_D}
        =
        R(\rho_{AB}),
    \end{equation}
we have
    \begin{equation}
        \frac{\lambda_{D-1}}{\lambda_D}
        +2\sqrt{\frac{\lambda_{D-1}}{\lambda_D}}
        \leq
        R(\rho_{AB})+2\sqrt{R(\rho_{AB})}.
    \end{equation}
Therefore, if
    \begin{equation}
        R(\rho_{AB})^n
        >
        R(\rho_{AB})+2\sqrt{R(\rho_{AB})},
    \end{equation}
then the necessary APPT condition \eqref{in:abppt} is violated. Equivalently,
this violation occurs whenever
    \begin{equation}
        n >
        \frac{
        \ln\!\left(R(\rho_{AB})+2\sqrt{R(\rho_{AB})}\right)
        }{
        \ln R(\rho_{AB})
        } .
    \end{equation}
Thus $\rho_{AB}^{\otimes n}$ is not absolutely PPT, and since absolute
separability implies absolute PPT, we conclude that
$\rho_{AB}^{\otimes n}\notin AS$. The required number of copies decreases monotonically with \(R(\rho_{AB})\).
It diverges as \(R(\rho_{AB})\to 1^+\), i.e. as the state approaches the
maximally mixed state, while it approaches \(1\) as
\(R(\rho_{AB})\to\infty\).
\end{proof}

\end{document}

%% file: figures/path1.pdf_tex
\begingroup%
  \makeatletter%
  \providecommand\color[2][]{%
    \errmessage{(Inkscape) Color is used for the text in Inkscape, but the package 'color.sty' is not loaded}%
    \renewcommand\color[2][]{}%
  }%
  \providecommand\transparent[1]{%
    \errmessage{(Inkscape) Transparency is used (non-zero) for the text in Inkscape, but the package 'transparent.sty' is not loaded}%
    \renewcommand\transparent[1]{}%
  }%
  \providecommand\rotatebox[2]{#2}%
  \newcommand*\fsize{\dimexpr\f@size pt\relax}%
  \newcommand*\lineheight[1]{\fontsize{\fsize}{#1\fsize}\selectfont}%
  \ifx\svgwidth\undefined%
    \setlength{\unitlength}{372.64409583bp}%
    \ifx\svgscale\undefined%
      \relax%
    \else%
      \setlength{\unitlength}{\unitlength * \real{\svgscale}}%
    \fi%
  \else%
    \setlength{\unitlength}{\svgwidth}%
  \fi%
  \global\let\svgwidth\undefined%
  \global\let\svgscale\undefined%
  \makeatother%
  \begin{picture}(1,0.55372682)%
    \lineheight{1}%
    \setlength\tabcolsep{0pt}%
    \put(0,0){\includegraphics[width=\unitlength,page=1]{path1.pdf}}%
    \put(0.22487719,0.08635389){\color[rgb]{0,0,0}\makebox(0,0)[lt]{\lineheight{1.25}\smash{\begin{tabular}[t]{l}$\frac{\mathds{1}-\phi^{+ \Gamma}}{d_Ad_B-1}$\end{tabular}}}}%
    \put(-0.00223408,0.37808621){\color[rgb]{0,0,0}\makebox(0,0)[lt]{\lineheight{1.25}\smash{\begin{tabular}[t]{l}$(\frac{1}{d_Ad_B-1},...,\frac{1}{d_Ad_B-1},0)$\end{tabular}}}}%
    \put(0.41066215,0.2520373){\color[rgb]{0,0,0}\makebox(0,0)[lt]{\lineheight{1.25}\smash{\begin{tabular}[t]{l}$\frac{\mathds{1}}{d_Ad_B}$\end{tabular}}}}%
    \put(0.6123007,0.3347312){\color[rgb]{0.13333333,0.37254902,0.23529412}\makebox(0,0)[lt]{\lineheight{1.25}\smash{\begin{tabular}[t]{l}CAS\end{tabular}}}}%
    \put(0.83456713,0.49859128){\color[rgb]{0.38039216,0.37647059,0.92941176}\makebox(0,0)[lt]{\lineheight{1.25}\smash{\begin{tabular}[t]{l}AS\end{tabular}}}}%
    \put(0,0){\includegraphics[width=\unitlength,page=2]{path1.pdf}}%
    \put(0.53889897,0.00679117){\color[rgb]{0,0,0}\makebox(0,0)[lt]{\lineheight{1.25}\smash{\begin{tabular}[t]{l}$\tilde{\rho}$\end{tabular}}}}%
  \end{picture}%
\endgroup%